\pdfoutput=1
\documentclass[11pt]{article}

\usepackage[]{ACL2023}

\usepackage{times}
\usepackage{latexsym}
\usepackage{todonotes} % for comments

\usepackage[T1]{fontenc}
\usepackage[utf8]{inputenc}

\usepackage{microtype}

\usepackage{inconsolata}

\usepackage{graphicx}
\usepackage{amsmath}
\usepackage{amsthm}
\usepackage{amssymb}
\usepackage[english]{babel}
\newtheorem{theorem}{Lemma}
\usepackage{color,soul}
\usepackage{multirow}
\usepackage{caption}
\usepackage{subcaption}
\usepackage{floatrow}
\usepackage{graphicx}
\usepackage{float}
\restylefloat{table}

\title{Modeling Relational Patterns for Logical Query Answering over Knowledge Graphs}

\author{Yunjie He \and Mojtaba Nayyeri \and Bo Xiong \and Yuqicheng Zhu \\University of Stuttgart \\ \texttt{\{yunjie.he, mojtaba.nayyeri, bo.xiong, yuqicheng.zhu \}@ipvs.uni-stuttgart.de} 
\AND
Evgeny Kharlamov \\ Bosch Center for Artificial Intelligence \\ \texttt{Evgeny.Kharlamov@de.bosch.com}
\And
Steffen Staab \\ University of Stuttgart  and\\ University of Southampton\\ \texttt{steffen.staab@ipvs.uni-stuttgart.de}
}

\begin{document}
\maketitle

\begin{abstract}
Answering first-order logical (FOL) queries over knowledge graphs (KG) remains a challenging task mainly due to KG incompleteness. 
Query embedding approaches this problem by computing the low-dimensional vector representations of entities, relations, and logical queries. 
KGs exhibit relational patterns such as symmetry and composition and modeling the patterns can further enhance the performance of query embedding models.
However, the role of such patterns in answering FOL queries by query embedding models has not been yet studied in the literature.
In this paper, we fill in this research gap and empower FOL queries reasoning with pattern inference by introducing an inductive bias that allows for learning relation patterns. 
To this end, we develop a novel query embedding method, RoConE, that defines query regions as geometric cones and algebraic query operators by rotations in complex space. RoConE combines the advantages of Cone as a well-specified geometric representation for query embedding, and also the rotation operator as a powerful algebraic operation for pattern inference. 
%Therefore, RoConE enables inferring patterns during the multi-hop reasoning process.
Our experimental results on several benchmark datasets confirm the advantage of relational patterns for enhancing logical query answering task.
\end{abstract}

\section{Introduction}
% Introduction to KG and link prediction task
% KGs (KGs) such as Wikidata \citep{Wikidata} and Freebase \citep{Freebase}
% %, and YAGO \citep{YAGO} 
% represent real-world facts as a set of triplets of the form (head entity, relation, tail entity). 
%Representing KGs in a low-dimensional space has been a fundamental task. Many popular embedding methods \citep{TransE, RotatE, TransH, ComplEx, DistMult} have been introduced to properly preserve the semantic and structural information in KGs. All of these embedding methods are trained with respect to the link prediction task, i.e. prediction of missing links between KG entities.
% Introduction to logical query answering task, what it is (example) and why it is important
%Beyond the link prediction task over incomplete KGs (KGs), 
Answering first-order logical (FOL) queries over knowledge graphs (KGs) has been an important and challenging problem \cite{BetaE}. 
% The problem itself is challenging because KGs are incomplete. 
%The methods which are based on traversing the large-scale KGs  exponential computational cost of traversing the large-scale KGs for reasoning.
Among various approaches,
logical query embedding \citep{GQE,Query2box,ConE,BetaE} has received huge attention due to its great efficiency and effectiveness. 
Logical query embeddings take as input the KGs and a set of first-order logical queries that include existential quantification ($\exists$), conjunction ($\land$), disjunction ($\lor$), and negation ($\neg$). 
%They then compute the vector representations of KGs and queries to measure if a candidate answer is plausible to be the answer to a query. 
%Such queries are usually named first-order logical (FOL) queries or complex logical queries. 
Figure \ref{fig:FOL query example} shows a concrete FOL query example that corresponds to the natural language query “List the capitals of non-European countries that have held either World Cup or Olympics”. 
%% Temporary example
\begin{figure}[t]
    \centering
    \includegraphics[width=\linewidth]{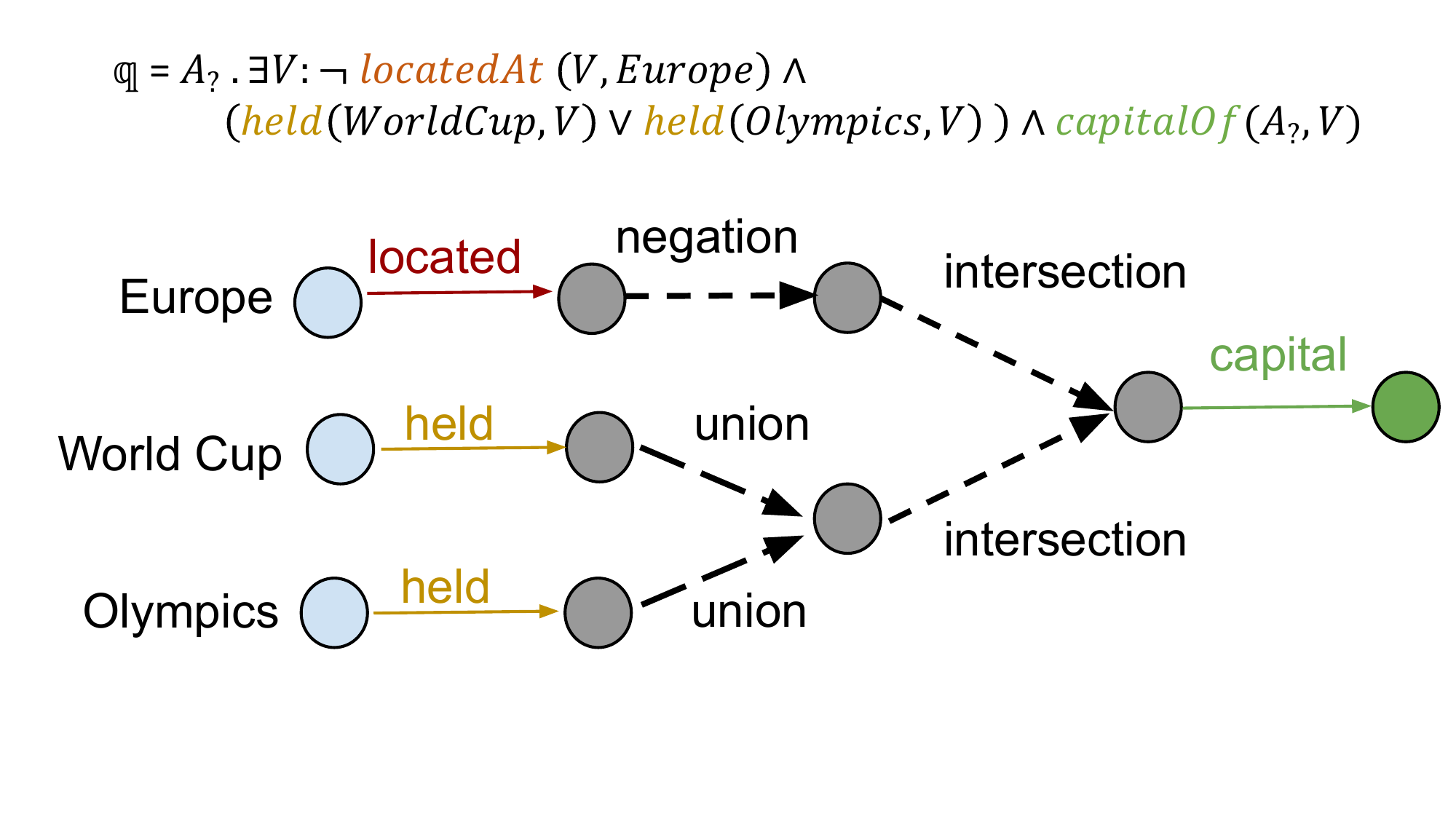}
    \caption{An example of FOL query corresponds to "List the capitals of non-European countries that have held either World Cup or Olympics".}
    \label{fig:FOL query example}
\end{figure}
% Summarize existing methods: traditional methods and drawbacks (KG incompleteness)
% Summarize existing methods: what are query embedding methods? why do they work in the incomplete graphs? what is the research gap?
%Answering first-order logical queries on KGs is an important yet challenging task mainly due to the incompleteness of KGs. Even for the most large-scale KGs, it is infeasible to include every real-world fact in them. Thus, reasoning in the presence of missing information is non-trivial for the existing query-answering models. 
%To tackle complex queries on incomplete KGs, 
The methods model logic operations by neural operators that act in the vector space. In particular, they represent a set of entities as geometric shapes and design neural-network-based logical operators to compute the embedding of the logical queries. 
The similarity between the embedded logical query and a candidate answer is calculated to measure plausibility.
%The geometric query embedding methods achieved great success in inventing predictions given the incompleteness of KGs and endowing logical operations with geometric explanations.
%\Mojtaba{I checked until here ...}

%KG relation patterns are also known as KG connectivity patterns which are modeled by existing embedding methods. 
Relations in KGs may form particular patterns, e.g.,
some relations are symmetric (e.g., spouse) while others are anti-symmetric (e.g., parent$\_$of); some relations are the inverse of other relations (e.g., son$\_$of and father$\_$of). %and some relations may be composed by others (e.g., my mother’s husband is my father). 
Modeling relational patterns can potentially improve the generalization capability and has been extensively studied in link prediction tasks \cite{RotatE,nayyeri2021losspattern}. 
This has been shown particularly in \citet{RotatE} where the RotatE model utilizes the rotation operation in Complex space to model relational patterns and enhance link prediction task. 
However, current logical query embedding models adopt deep neural logical operators and are not able to model relational patterns, which do matter for logical query answering. Figure \ref{relation pattern} shows concrete examples of how relation patterns can impact complex query reasoning in KGs. 
\begin{figure*}[t]
    \centering
    \includegraphics[scale=0.5]{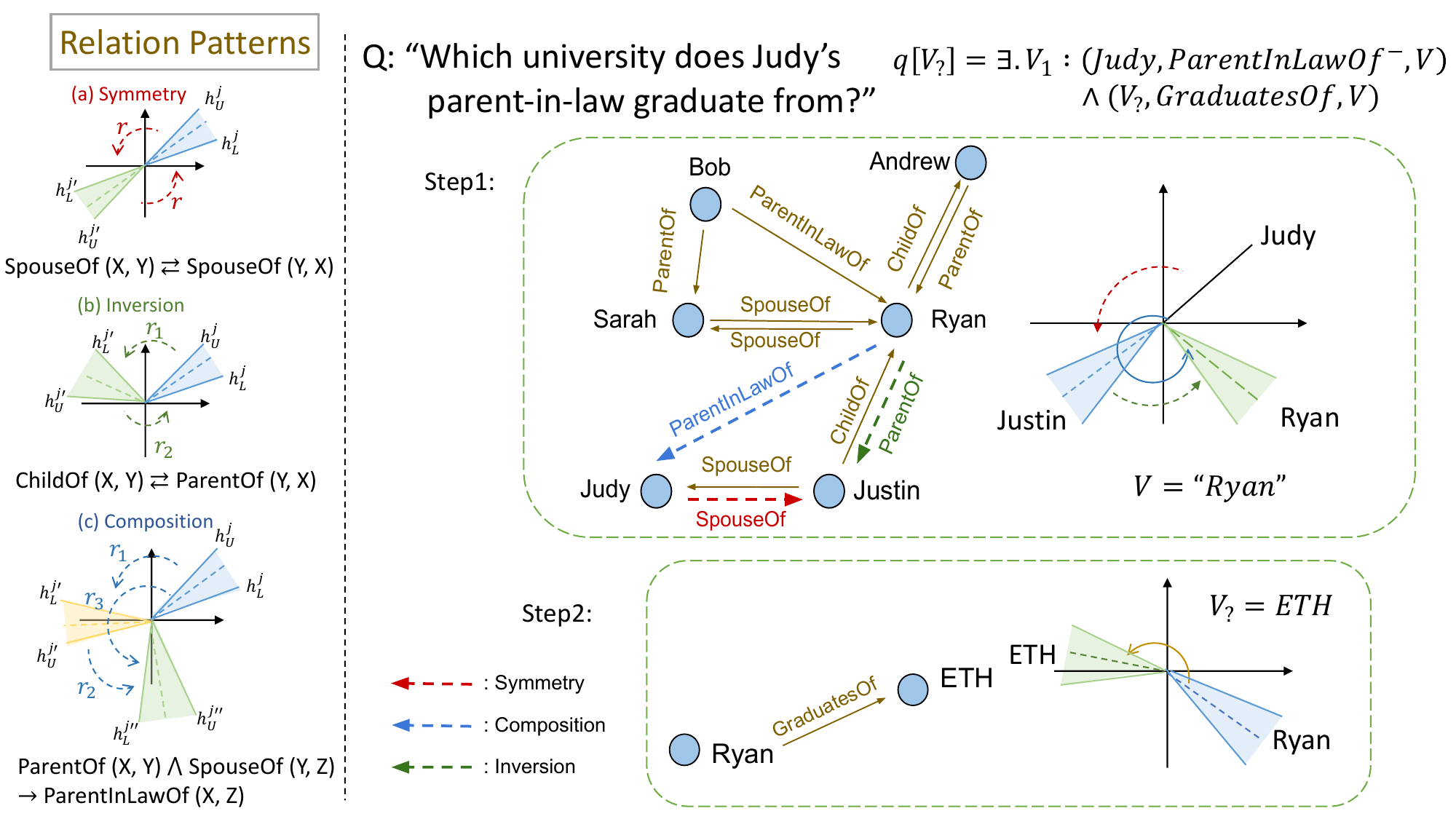}
    \caption{(\emph{top}) An example showing how relation patterns influence query answering over incomplete KGs: 
    the intermediate variable Judy's father-in-law in the query cannot be directly extracted from the given facts;
    (\emph{bottom}) An illustration of cone rotation. Based on the existing relations between Judy, Justin, and Ryan, and the learned potential relation patterns from other parts of the graph, the model is able to derive the following information by relational rotation: (i) Justin is Judy's spouse (symmetric rotation) (ii) Justin is the child of Ryan (inversion rotation) (iii) Ryan is Judy's father in law (compositional rotation). With the predicted query embedding on $V$, the model is able to derive where $V$ graduated from by another relational rotation.}
    \label{relation pattern}
\end{figure*}

% \begin{figure}
% \begin{subfigure}
% \centering
% \includegraphics[width=0.2\textwidth]{Pictures/left.pdf}
% \caption{$y=x$}
% \label{fig:y equals x}
% \end{subfigure}
% \begin{subfigure}
% \centering
% \includegraphics[width=0.2\textwidth]{Pictures/center.pdf}
% \caption{$y=3\sin x$}
% \label{fig:three sin x}
% \end{subfigure}
% \hfill
% \begin{subfigure}
% \centering
% \includegraphics[width=0.2\textwidth]{Pictures/right.pdf}
% \caption{$y=5/x$}
% \label{fig:five over x}
% \end{subfigure}
% %\includegraphics[scale=0.4]{Pictures/relation patterns.pdf}
% \caption{Examples on how relation patterns influence query answering over incomplete KGs. The left panel illustrates the rotations corresponding to different relation patterns on 2D complex plane. The right panel displays how the captured relation patterns help with reasoning a two-hop query given an incomplete KG: the first step is to find out who is Judy's father-in-law, which cannot be directly extracted from the given facts. However, based on the learned potential relation patterns from other parts of the graph and the existing relations between Judy, Justin, and Ryan, the model is able to derive that Ryan is Judy's father-in-law through relational rotation in complex space. The second step is simply to derive where Ryan graduated from by another relational rotation. }
% \label{relation pattern}
% \end{figure}

To model and infer various KG relational patterns in logical queries answering process, we propose a novel method called RoConE  that combines the advantages
of Cone as a well-specified geometric representation for query embedding \cite{ConE}, and also the
rotation operator \cite{RotatE} as a powerful algebraic operation for pattern inference.
We define each relation as a rotation from the source entity set to the answer/intermediate entity set and perform neural logical operators upon the selected entity sets in the complex vector space. We provide theoretical proof of its ability to model relational patterns, as well as experimental results on how the relational patterns influence logical queries answering over three established benchmark datasets.
% Introduction to our contributions: 

%In this paper, we primarily make the following contribution:

\section{Related Work}
%\paragraph{KG Relation Patterns}
%An overriding goal in KG embedding methods has been to learn KG representation while capturing distinct latent logical/relation patterns (symmetry/anti-symmetry, composition, and inversion). A plethora of approaches has been presented to generate representations of KGs \citep{TransE, DistMult, ComplEx, RotatE, TransH}. Except for RotatE, which can capture all of them in complex space, most of them can explicitly or implicitly model one or more of the aforementioned patterns. However, all of these KG embedding methods only focus on the link prediction task or single-hop queries but are limited in answering multi-hop and first-order logical queries.
%\paragraph{Complex query answering}
%Here we review the query answering methods as follows.
To answer more complex queries, a number of path-based \citep{deeppath,lin-etal-2018-multi}, neural \citep{GQE,Query2box,BiQE}, and neural-symbolic \citep{CQD,GNNQE} methods have been developed. Among these methods, geometric and probabilistic query embedding approaches \citep{GQE,Query2box,ConE,BetaE} provide a way to tractably handle first-order logic operators in queries and equip excellent computing efficiency. This is done by representing entity sets as geometric objects or probability distributions, such as boxes \citep{Query2box}, cones \citep{ConE}, or Beta distribution \citep{BetaE}, and performing neural logical operations directly on them. In this way, the expensive search for intermediate variables in multi-hop queries is avoided. %The Graph Query Embedding (GQEs) \citep{GQE} was first proposed to answer only conjunctive queries via modeling the query as single vector $\bf{q}$ through geometric translational operators. However, representing a query as a single vector limits the model's expressiveness in modeling multiple entities. Query2Box \citep{Query2box} remedies this flaw by modeling entities as points within boxes. This allows Q2B to predict the intersection of entity sets as the intersection of boxes in vector space. 
%ConE \citep{ConE} is proposed as the first geometry-based query embedding method that can handle complete FOL queries via embedding the set of entities (query embedding) as cones in Euclidean space.% 
All of the above query embedding methods commonly apply multi-layer perceptron networks for selecting answer entities of atomic queries by relation and performing logical operations. Thus, their ability to capture relation patterns in KGs remains unclear. Our proposed method RoConE fills in this gap and combines the benefits of both worlds (KG embedding and complex query answering) together.

\section{Preliminaries}
\paragraph{Knowledge Graph}
 A KG $\mathcal{G} \subseteq \mathcal{E} \times \mathcal{R} \times \mathcal{E}$, where $\mathcal{E}$ and $\mathcal{R}$ represent the set of entities and relations respectively, can be defined as a set of subject-predicate-object $\bigl\langle s,p,o \bigr\rangle$ triples. For each triple $\bigl\langle e_i,r,e_j \bigr\rangle$, $e_{i,j} \in \mathcal{E}$ and $r \in \mathcal{R}$, it exists if and only if $e_i$ is linked to $e_j$ by relation $r$.

% \paragraph{Relation patterns}
% Going beyond triple representation, KGs also exhibit relational patterns. We introduce the three most common patterns as follows:\\
% 1. Symmetry/anti-symmetry: 
% A relation $r$ is symmetric (anti-symmetric) if
% \(\forall x, y: r(x,y) \rightarrow r(y,x) (r(x,y) \rightarrow \neg r(y,x) )\)\\
% 2. Inversion: relation $r_1$ is inverse to relation $r_2$ if
% \(\forall x, y: r_{1}(x,y) \rightarrow r_{2}(y,x)\)\\
% 3. Composition: relation $r_3$ is  composition of relation $r_1$ and $r_2$ if 
% \(\forall x, y, z: r_{1}(x,y) \land r_{2}(y,z) \rightarrow r_{3}(x,z)\)

% first order logical queries
\paragraph{First-Order Logical Queries involving constants}
First-Order Logical Queries are broad and here we consider answering a subset, i.e., 
multi-hop queries with constants and first-order logical operations including conjunction ($\land$), disjunction ($\lor$), existential quantification ($\exists$), and negation ($\neg$). 
The query consists of a set of constants (anchor entities) $\mathcal{E}_a \subset \mathcal{E}$, a set of existentially quantified bound variables $V_1,..., V_m$ and a single target answer variable $V_?$. The disjunctive normal form of this subset of FOL queries is namely the disjunction of conjunctive formulas, and can be expressed as
\begin{equation}
    q[V_t] = V_t .\exists V_1,...,V_m: c_1 \lor c_2 \lor ... \lor c_n
\end{equation}
where $c_i$, $i\in \{1,...,n\}$ corresponds to a conjunctive query with one or more atomic queries $d$ i.e. \(c_i= d_{i1} \land d_{i2}\land ... \land d_{im}\). For each atomic formula, $d_{ij}$ = $(e_a,r,V)$ or $\neg (e_a,r,V)$ or $(V^{'},r,V)$ or $\neg (V^{'},r,V)$, where $e_a \in \mathcal{E}_a$, $V \in \{V_t, V_1, ..., V_k\}$, $V^{'} \in \{V_1, ..., V_k\}$, $r\in \mathcal{R}$. 
The goal of logical query embedding is to find a set of answer entities 
$\{e_{t1},e_{t2},...\}$ for $V_t$, such that $q[V_t]=\operatorname{True}$.

% Logical operators
% \paragraph{Logical Operators}
% As illustrated in Figure \ref{fig:FOL query example}, each logical query can be represented as a directed acyclic graph (DAG) tree, where the tree nodes correspond to constants/anchor node entities or variables, and the edges correspond to atom relations or logical operations in a query. Logical operations are performed along the DAG tree from constants to the target answer variable. Following \citep{BetaE}, the logical operators can be defined as:\\
% 1. Relational Projector: Given a set of entities $\mathcal{S} \subset \mathcal{E}$ and relation $r \in \mathcal{R}$, selects the neighbouring entities $\mathcal{S}^{'}\subset\mathcal{E}$ such that $\mathcal{S}^{'} = \{e \in \mathcal{S}, e^{'} \in \mathcal{S}^{'}: r(e, e^{'}) = True\}$\\
% 2. Intersector: For sets of entities $\{S_{1},...,S_{n}\}$, their intersection is $\cap_{i=1}^{n} S_{i}$\\
% 3. Union: Given sets of entities $\{S_{1},...,S_{n}\}$, their union is $\cup_{i=1}^{n} S_{i}$\\
% 4. Negation: Given a set of entities $\mathcal{S} \subset \mathcal{E}$, its negation can be defined as $\bar{S} = \mathcal{E}\setminus\mathcal{S}$

\section{Methodology}
To accommodate the learning of relation patterns in query answering over KGs, we propose a new model RoConE, which models entity sets as cones and relations as anti-clockwise angular rotations on cones in the complex plane. Each cone $\mathbf{q}$ is parameterized by $\mathbf{q} =(\mathbf{h}_{U}, \mathbf{h}_{L})$, where $|\mathbf{h}_{\{U,L\}}|=\textbf{1}$, and
$\mathbf{h}_{U}$, $\mathbf{h}_{L} \in \mathbb{C}^d$ represent the counter-clockwise upper and lower boundaries of a cone, such that 
\begin{equation}
\begin{split}
    \mathbf{h}_{U} \equiv e^{i\pmb{\theta}_{U}} \equiv e^{i(\pmb{\theta}_{ax} + \pmb{\theta}_{ap}/2)}, \\
    \mathbf{h}_{L} \equiv e^{i\pmb{\theta}_{L}} \equiv e^{i(\pmb{\theta}_{ax} - \pmb{\theta}_{ap}/2)},
\end{split}
\end{equation}
where $\pmb{\theta}_{ax} \in [-\pi,\pi)^d$ represents the angle between the symmetry axis of the cone and $\pmb{\theta}_{ap} \in [0,2\pi]^d$ represents the cone aperture, and $d$ is the embedding dimension.
The query and the set of entities are modeled as cones and each entity instance is modeled as a vector such that $\mathbf{h}_{U} = \mathbf{h}_{L}$.

\subsection{Logical Operators}
As illustrated in Figure \ref{fig:FOL query example}, each logical query can be represented as a directed acyclic graph (DAG) tree, where the tree nodes correspond to constants/anchor node entities or variables, and the edges correspond to atom relations or logical operations in a query. Logical operations are performed along the DAG tree from constants to the target answer variable. Figure \ref{loigcal operators} visualizes these operations on 2D complex space.
The logical operators can be defined as follows

\paragraph{Relational rotating projection}
Given a set of entities $\mathcal{S} \subset \mathcal{E}$ and a relation $r \in \mathcal{R}$, the projection operator selects the neighbouring entities $\mathcal{S}^{'}\subset\mathcal{E}$ by relation such that $\mathcal{S}^{'} = \{e \in \mathcal{S}, e^{'} \in \mathcal{S}^{'}: r(e, e^{'}) = True\}$. 
Existing query embedding methods \citep{ConE,Query2box,GQE,BetaE} apply multi-layer perceptron networks to accomplish this task. They do not accommodate the learning of potential KG relational patterns which might help in reasoning logical queries. Motivated by RotatE \citep{RotatE}, we represent each relation $\mathbf{r}$ as a counterclockwise relational rotation on query embeddings about the origin of the complex plane such that $\mathbf{r} = (\mathbf{r}_{U}, \mathbf{r}_{L})$, where $|\mathbf{r}_{\{U,L\}}|=\textbf{1}$, and $\mathbf{r}_{U}, \mathbf{r}_{L} \in \mathbb{C}^{d}$. Given the query embedding $\mathbf{q} = (\mathbf{h}_{U}, \mathbf{h}_{L})$ and a relation $\mathbf{r}$, the selected query embedding $\mathbf{q}^{'} = (\mathbf{h}_{U}^{'}, \mathbf{h}_{L}^{'})$ is
\begin{equation}
    \begin{split}
    \mathbf{h}_{U}^{'} &= \mathbf{h}_{U} \circ \mathbf{r}_{U} \equiv e^{i(\pmb{\theta}_{ax} + \pmb{\theta}_{ax,r} + (\pmb{\theta}_{ap}+\pmb{\theta}_{ap,r})/2)}, \\
    \mathbf{h}_{L}^{'} &= \mathbf{h}_{L} \circ \mathbf{r}_{L} \equiv e^{i(\pmb{\theta}_{ax} + \pmb{\theta}_{ax,r} - (\pmb{\theta}_{ap}+\pmb{\theta}_{ap,r})/2)},  
    \end{split}
\end{equation}
where $\circ$ is the Hadmard (element-wise) product, and $\pmb{\theta}_{ax,r}$, $\pmb{\theta}_{ap,r}$ correspond to the equivalent relational rotation on $\pmb{\theta}_{ax}$, $\pmb{\theta}_{ap}$. Specifically, for each element of the cone embeddings, we have $h_{U,i}^{'} = h_{U,i}r_{U,i}$ and $h_{L,i}^{'} = h_{L,i}r_{L,i}$.
Each element $r_i$ of the relational rotation $\mathbf{r}_{\{U,L\}}$ corresponds to a counterclockwise rotation on the matching element of upper or lower boundaries by $\theta_{r,i}$ radians about the origin of the complex plane.
By modeling the projection as relational rotation on a cone in the complex space, RoConE is shown to model and infer all three types of relation patterns introduced in Section 3. 
The lemmas and their proofs are in Appendix and the rotations corresponding to different relation patterns are visualized in Figure \ref{relation pattern rotations}.

\paragraph{Intersection}
For the input cone embeddings of entity sets \(\{ \mathbf{q}_{1},...,\mathbf{q}_{n}\}\), the intersection operator selects the intersection $\mathbf{q}^{'} = \cap_{j=1}^{n} \mathbf{q}_{j}$ with the \textbf{SemanticAverage}\footnote{ \textbf{SemanticAverage} and \textbf{CardMin} are explained in Appendix \ref{Intersection Operator}}($\cdot$) and \textbf{CardMin}($\cdot$) \citep{ConE}, which calculate the semantic centers and apertures of cones respectively. Since we have each cone \(\mathbf{q}_{j} = (\mathbf{h}_{j,U}, \mathbf{h}_{j,L}) \equiv (e^{i(\pmb{\theta}_{j,ax}+\pmb{\theta}_{j,ap}/2)}, e^{i(\pmb{\theta}_{j,ax}-\pmb{\theta}_{j,ap}/2)})\), the intersection $\mathbf{q}^{'} = (\mathbf{h}_{U}^{'}, \mathbf{h}_{L}^{'})$ can be defined as follows
\begin{equation}
\begin{split}
     &\mathbf{h}_{U}^{'} = e^{i(\pmb{\theta}_{ax}^{'} +\pmb{\theta}_{ap}^{'}/2)},\\
     &\mathbf{h}_{L}^{'} = e^{i(\pmb{\theta}_{ax}^{'}
    -\pmb{\theta}_{ap}^{'}/2)},
 \end{split}
\end{equation}
where 

\begin{equation}
\begin{split}
    \pmb{\theta}_{ax}^{'} &= \textbf{SemanticAverage}(\{(\pmb{\theta}_{j,ax}, \pmb{\theta}_{j,ap})\}_{j=1}^{n}),\\
    \pmb{\theta}_{ap}^{'} &= \textbf{CardMin}(\{(\pmb{\theta}_{j,ax}, \pmb{\theta}_{j,ap})\}_{j=1}^{n}).
\end{split}
\end{equation}

\paragraph{Disjunction}
Given the input cone embeddings of entity sets \(\{ \mathbf{q}_{1},...,\mathbf{q}_{n}\}\) where $\mathbf{q}_{j} = (\mathbf{h}_{j,U}, \mathbf{h}_{j,L})$, the disjunction operator finds the union set $\mathbf{q}^{'} = \cup_{j=1}^{n} \mathbf{q}_{j} = \{\mathbf{q}_{1}, \dots, \mathbf{q}_{n}\}$, which is equivalent to
\begin{equation*}
    \resizebox{\hsize}{!}{$(\{(\text{h}_{1,U}^{1}, \text{h}_{1,L}^{1}),\dots,(\text{h}_{n,U}^{1}, \text{h}_{n,L}^{1})\},\dots,\{(\text{h}_{1,U}^{d}, \text{h}_{1,L}^{d}),\dots,(\text{h}_{n,U}^{d}, \text{h}_{n,L}^{d})\})$}
\end{equation*}
Following \citet{Query2box}, we also adopt DNF technique to translate FOL queries into the disjunction of conjunctive queries and only perform the disjunction operator in the last step in the computation graph.
\paragraph{Negation}
Given a set of entities $\mathcal{S} \subset \mathcal{E}$, the negation operator finds its complementary negation $\bar{S} = \mathcal{E}\setminus\mathcal{S}$. Given the cone embedding of entity set $\mathcal{S}$, $\mathbf{q}^{\mathcal{S}} = (\mathbf{h}_{U}^{\mathcal{S}}, \mathbf{h}_{L}^{\mathcal{S}})$, its corresponding complementary negation \(\bar{\mathbf{q}} = (\mathbf{h}_{L}^{\mathcal{S}}, \mathbf{h}_{U}^{\mathcal{S}})\).

\subsection{Optimization}
Given a set of training samples, our goal is to minimize the distance between the query cone embedding $\mathbf{q} = (\mathbf{h}_{U}^{q}, \mathbf{h}_{L}^{q})$ and the answer entity vector $\mathbf{h}^{*}$ while maximizing the distance between this query and negative samples. Thus, we define our training objective, the negative sample loss as
\begin{equation}
    \resizebox{\hsize}{!}{
    $L = -log \sigma (\gamma - d(\mathbf{h}^{*},\mathbf{q})) - \frac{1}{k}\sum_{i=1}^{k}log\sigma(d(\mathbf{h}^{*}_{i},\mathbf{q})-\gamma)
    $}
\end{equation}
where $d(\cdot)$ is the combined distance specifically defined in Appendix \ref{Distance function}, $\gamma$ is a margin, $\mathbf{h}^{*}$ is a positive entity and $\mathbf{h}^{*}_{i}$ is the i-th negative entity, k is the number of negative samples, and $\sigma(\cdot)$ represents the sigmoid function.

\section{Experiments}
\begin{table}
\centering
\resizebox{\hsize}{!}{
\begin{tabular}{lllllllllll}
\hline
 \textbf{Dataset} & \textbf{Model} & \textbf{1p} & \textbf{2p} & \textbf{3p} & \textbf{2i} & \textbf{3i} & \textbf{pi} & \textbf{ip} & \textbf{2u} & \textbf{up}\\
\hline
 & GQE & 35.2 & 7.4	& 5.5 & 23.6 & 35.7 & 16.7 & 10.9 & 8.4 & 5.8\\
 & Query2Box & 41.3 & 9.9 & 7.2 & 31.1 & 45.4 & 21.9 & 13.3 & 11.9 & 8.1\\
FB15k-237 & BetaE & 39.0 & \underline{10.9} & \underline{10.0} & 28.8 & 42.5 & 22.4 & 12.6 & 12.4 & \underline{9.7}\\
 & ConE & 41.8 & \textbf{12.8} & \textbf{11} & \underline{32.6} & 47.3 & \textbf{25.5} & 14.0 & \textbf{14.5} & \textbf{10.8}\\
 & RoConE & \textbf{42.2} & 10.5 & 7.5 & \textbf{33.5} & \textbf{48.1} & \underline{23.5} & \textbf{14.5} & \underline{12.8} & 8.9\\
\hline
 & GQE & 33.1 & 12.1 & 9.9 & 27.3 & 35.1 & 18.5 & 14.5 & 8.5 & 9.0\\
 & Query2Box & 42.7 & 14.5 & 11.7 & 34.7 & 45.8 & 23.2 & 17.4 & 12.0 & 10.7\\
NELL995 & BetaE & 53.0 & 13.0 & 11.4 & 37.6 & 47.5 & 24.1 & 14.3 & 12.2 & 8.5\\
 & ConE & \underline{53.1} & \underline{16.1} & \underline{13.9} & \underline{40.0} & \underline{50.8} & \textbf{26.3} & \underline{17.5} & \underline{15.3} & \underline{11.3}\\
 & RoConE & \textbf{54.5} & \textbf{17.7} & \textbf{14.4} & \textbf{41.9} & \textbf{53.0} & \underline{26.1} & \textbf{20.7} & \textbf{16.5} & \textbf{12.8}\\
\hline
\end{tabular}
}
\caption{
MRR results (\%) of RoConE, BETAE, Q2B, and GQE on answering EPFOL ($\exists, \land, \lor$) queries. The best statistic is highlighted in bold, while the second best is highlighted in underline.
}
\label{EPFOL}
\end{table}

\paragraph{Experiment setup}
We evaluate RoConE on two benchmark datasets NELL995 \citep{NELL995} and FB15k-237 \citep{FB15k-237}. RoConE is compared with various state-of-the-art query embedding models. Mean reciprocal rank (MRR) is used as the metric. More experimental details are in Appendix \ref{Experiment}. 

\paragraph{Main results}
Table \ref{EPFOL} summarizes the performance of all methods on answering various query types without negation. 
% compared with baselines that can only model EPFOL (Existential Positive First-Order Logical) with conjunction and disjunction without negation. 
RoConE outperforms baseline methods on the majority of query types while achieving competitive results on the others. We also observed that RoConE shows better performances on NELL-995 than those on FB15k-237. We conjecture that this is due to the discrepancy in the distribution of relation patterns between these two datasets.
% It demonstrates that RoConE outperforms the state-of-the-art query embedding models on both datasets across the majority of EPFOL query types. It is noticed that RoConE shows better performances on reasoning multi-hop queries from NELL995 than those from FB15k-237. We suspect this is due to the discrepancy in the distribution of relation patterns between these two datasets.
As Table \ref{Negation results} shows, RoConE does not bring many improvements for answering query types involving negation. There are two folds of possible reasons that might lead to this result. Firstly, the traditional modeling of negation as complements may be problematic, which can be reflected in the poor performance of all existing QE models. This largely brings too much uncertainty into the query embedding and leads to severe bias in prediction. Secondly, the influence of relation patterns on negation queries is limited when we model negation as complements. % On the other hand, the learning of relation patterns does not bring improvements to the queries with negation. There are two folds of possible reasons that might lead to this result. Firstly, the traditional modeling of negation as complements may be problematic, at least, this largely brings too much uncertainty into the query embedding, which leads to severe bias in prediction. Secondly, the influence of relation patterns on negation queries is limited when we model negation as complements. 

\begin{table}[ht]
\centering
\resizebox{\hsize}{!}{
\begin{tabular}{lllllll}
\hline
 \textbf{Dataset} & \textbf{Model} & \textbf{2in} & \textbf{3in} & \textbf{inp} & \textbf{pin} & \textbf{pni}\\
\hline
 \multirow{3}{5em}{FB15k-237} & BetaE & 5.1 & 7.9 & 7.4 & 3.6 & 3.4\\
 & ConE & \textbf{5.4} & \textbf{8.6} & \textbf{7.8} & \textbf{4.0} & \textbf{3.6}\\
 & RoConE & 4.1 & 7.9 & 6.9 & 3.1 & 2.8\\
\hline
\multirow{3}{4em}{NELL995} & BetaE & 5.1 & 7.8 & 10 & 3.1 & 3.5\\
 & ConE & \textbf{5.7} & \textbf{8.1} & \textbf{10.8} & \textbf{3.5} & \textbf{3.9}\\
 & RoConE & 5.2 & 7.7 & 9.4 & 3.2 & 3.7\\
\hline
\end{tabular}
}
\caption{
MRR results (\%) of RoConE, BETAE, and ConE on answering queries with negation ($\neg$).
}
\label{Negation results}
\end{table}

\paragraph{Ablation study}
To investigate the influence of relation patterns on the query answering model, we designed an ablation study for RoConE on NELL995. The results are reported in Table \ref{ablation study}. RoConE (Base) denotes the neural baseline model without the relational rotating projection module. RoConE (S.E) and RoConE (Trunc) correspond to two variants of RoConE with different rotation strategies. More details are elaborated in Appendix \ref{Variants of RoConE}. The overperformance of RoConE and RoConE (truncation) reconfirms the efficiency of relation patterns in logical query reasoning tasks.

\begin{table}[ht]
\centering
\resizebox{\hsize}{!}{
\begin{tabular}{llllllllll}
\hline
 \textbf{Model} & \textbf{1p} & \textbf{2p} & \textbf{3p} & \textbf{2i} & \textbf{3i} & \textbf{pi} & \textbf{ip} & \textbf{2u} & \textbf{up}\\
\hline
RoConE (Base) & 53.1 & 16.1 & 13.9 & 40 & 50.8 & 26.3 & 17.5 & 15.3 & 11.3\\
RoConE (S.E.) & 51.3 & 16.6 & 13.8 & 38.4 & 48.4 & 18.9 & 19.5 & 14.7 & 12.1\\
RoConE (Trunc) & 53.7 & \textbf{17.8} & 14.2 & \textbf{41.9} & \textbf{53.0} & \textbf{27.5} & 20.4 & \textbf{16.6} & 12.7 \\
RoConE & \textbf{54.5} & 17.7 & \textbf{14.4} & \textbf{41.9} & \textbf{53.0} & 26.1 & \textbf{20.7} & 16.5 & \textbf{12.8}\\
\hline
\end{tabular}
}
\caption{
Ablation study of RoConE on NELL995
}
\label{ablation study}
\end{table}

\section{Conclusion}
In this paper, we theoretically and experimentally investigate the influence of relation patterns on enhancing the logical query reasoning task. By combining the relational rotating projection with the cone query embedding model in complex space, we improve FOL queries reasoning with relation pattern inference.

\section{Limitations}
RoConE incorporates the learning of potential KG relation patterns into the existing query embedding model for solving logical queries. We provide initial proof of the efficiency of relation patterns to complex reasoning via both theoretical explanations and experiments. One limitation of RoConE is that the relational rotating projection can not be generalized to other geometric query embedding methods, except for cone embedding, due to the restriction of natural geometry properties. For future work, we will propose more general and effective strategies to enhance the learning of relation patterns in the complex query reasoning task. 

\section{Ethics Statement}
The authors declare that we have no conflicts of interest. This article does not contain any studies involving business data and personal information.

\bibliography{anthology,custom}
\bibliographystyle{acl_natbib}
\newpage
\section*{Appendices}
\addcontentsline{toc}{section}{Appendices}
\renewcommand{\thesubsection}{\Alph{subsection}}

%\subsection{Connecting Relational Rotation Between Different Parameterizations}
\begin{figure*}[ht]
    \centering
    \includegraphics[scale=0.6]{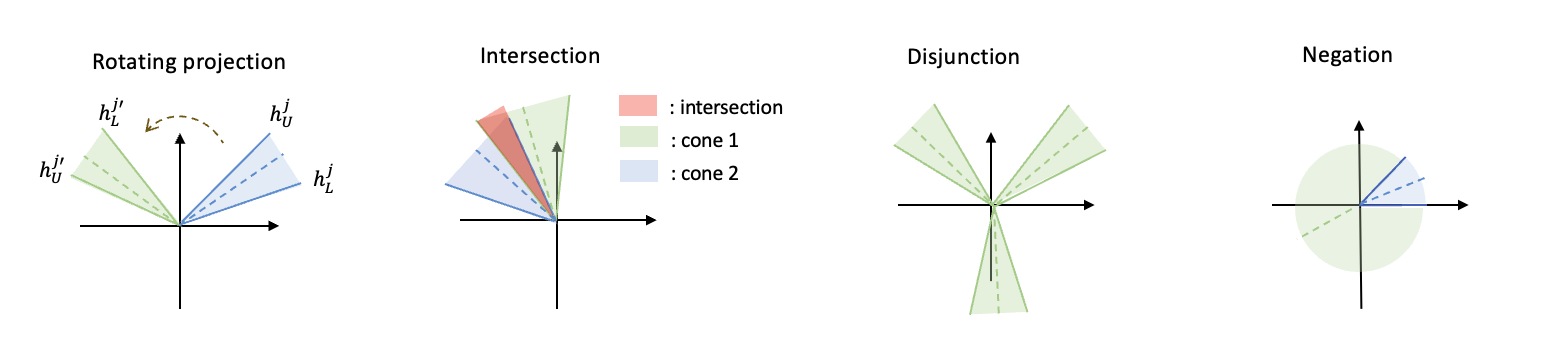}
    \caption{Visualization of logical operators \citep{ConE}}
    \label{loigcal operators}
\end{figure*}

\begin{figure*}[ht]
    \centering
    \includegraphics[scale=0.5]{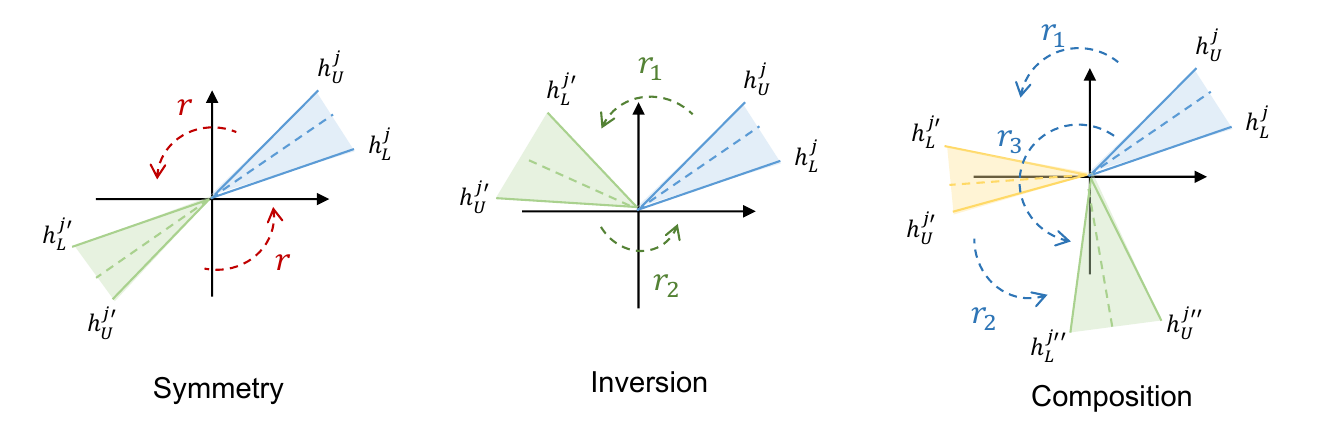}
    \caption{Relational rotations corresponding to different relation patterns on 2D complex plane.}
    \label{relation pattern rotations}
\end{figure*}

\subsection{Lemmas and Corresponding Proofs}
\begin{theorem}
RoConE can infer the symmetry and anti-symmetry patterns. 
\end{theorem}
\begin{proof}
Given cone query embeddings $X = (\mathbf{h}_{U}^{x},\mathbf{h}_{L}^{x})$, $Y = (\mathbf{h}_{U}^{y},\mathbf{h}_{L}^{y})$ and relation $r = (\mathbf{r}_{U},\mathbf{r}_{L})$,
if $(X,r,Y)$ and $(Y,r,X)$ hold, we have:
\begin{equation}
\begin{split}
    (\mathbf{h}_{U}^{x}\circ\mathbf{r}_{U},\mathbf{h}_{L}^{x}\circ\mathbf{r}_{L}) &=(\mathbf{h}_{U}^{y},\mathbf{h}_{L}^{y}),\\
    (\mathbf{h}_{U}^{y}\circ\mathbf{r}_{U},\mathbf{h}_{L}^{y}\circ\mathbf{r}_{L}) &=(\mathbf{h}_{U}^{x},\mathbf{h}_{L}^{x}),
\end{split}
\end{equation}
\[\rightarrow (\mathbf{h}_{U}^{y}\circ\mathbf{r}_{U}\circ\mathbf{r}_{U},\mathbf{h}_{L}^{y}\circ\mathbf{r}_{L}\circ\mathbf{r}_{L})=(\mathbf{h}_{U}^{y},\mathbf{h}_{L}^{y})\]
thus, it is inferrable that $X$ is symmetric to $Y$ if and only if $(\mathbf{r}_{U}\circ\mathbf{r}_{U} = \mathbf{1},\mathbf{r}_{L}\circ\mathbf{r}_{L} = \mathbf{1})$. \\
Similarly, if $(X,r,Y)$ and $\neg(Y,r,X)$ hold, we have:
\begin{equation}
\begin{split}
        (\mathbf{h}_{U}^{x}\circ\mathbf{r}_{U},\mathbf{h}_{L}^{x}\circ\mathbf{r}_{L})&=(\mathbf{h}_{U}^{y},\mathbf{h}_{L}^{y}),\\(\mathbf{h}_{U}^{y}\circ\mathbf{r}_{U},\mathbf{h}_{L}^{y}\circ\mathbf{r}_{L})&\neq(\mathbf{h}_{U}^{x},\mathbf{h}_{L}^{x}),
\end{split}
\end{equation}
\[\rightarrow (\mathbf{h}_{U}^{x}\circ\mathbf{r}_{U}\circ\mathbf{r}_{U},\mathbf{h}_{L}^{x}\circ\mathbf{r}_{L}\circ\mathbf{r}_{L})\neq(\mathbf{h}_{U}^{x},\mathbf{h}_{L}^{x})\]
thus, it is inferrable that $X$ is anti-symmetric to $Y$ if and only if $(\mathbf{r}_{U}\circ\mathbf{r}_{U} \neq \mathbf{1},\mathbf{r}_{L}\circ\mathbf{r}_{L} \neq \mathbf{1})$. 
\end{proof}

\begin{theorem}
RoConE can infer the inversion pattern.
\end{theorem}
\begin{proof}
Given cone query embeddings $X = (\mathbf{h}_{U}^{x},\mathbf{h}_{L}^{x})$, $Y = (\mathbf{h}_{U}^{y},\mathbf{h}_{L}^{y})$ and relations $r_{1} = (\mathbf{r}_{U,1},\mathbf{r}_{L,1})$, $r_{2} = (\mathbf{r}_{U,2},\mathbf{r}_{L,2})$
if $(X,r_{1},Y)$ and $(Y,r_{2},X)$ hold, we have:
\begin{equation}
\begin{split}
        (\mathbf{h}_{U}^{x}\circ\mathbf{r}_{U,1},\mathbf{h}_{L}^{x}\circ\mathbf{r}_{L,1})&=(\mathbf{h}_{U}^{y},\mathbf{h}_{L}^{y}),\\(\mathbf{h}_{U}^{y}\circ\mathbf{r}_{U,2},\mathbf{h}_{L}^{y}\circ\mathbf{r}_{L,2})&=(\mathbf{h}_{U}^{x},\mathbf{h}_{L}^{x}),
\end{split}
\end{equation}
\[\rightarrow (\mathbf{h}_{U}^{y}\circ\mathbf{r}_{U,2}\circ\mathbf{r}_{U,1},\mathbf{h}_{L}^{y}\circ\mathbf{r}_{L,2}\circ\mathbf{r}_{L,1})=(\mathbf{h}_{U}^{y},\mathbf{h}_{L}^{y})\]
thus, it is inferrable that X is inversible to Y if and only if $\mathbf{r}_{U,2}\circ\mathbf{r}_{U,1}=1$ and $\mathbf{r}_{L,2}\circ\mathbf{r}_{L,1}=1$.
\end{proof}

\begin{theorem}
RoConE can infer the composition pattern.
\end{theorem}
\begin{proof}
Given cone query embeddings $X = (\mathbf{h}_{U}^{x},\mathbf{h}_{L}^{x})$, $Y = (\mathbf{h}_{U}^{y},\mathbf{h}_{L}^{y})$, $Z = (\mathbf{h}_{U}^{z},\mathbf{h}_{L}^{z})$ and relations $r_{1} = (\mathbf{r}_{U,1},\mathbf{r}_{L,1})$, $r_{2} = (\mathbf{r}_{U,2},\mathbf{r}_{L,2})$
if $(X,r_{1},Y)$ and $(Y,r_{2},Z)$ hold, we have:
\begin{equation}
    \begin{split}
    (\mathbf{h}_{U}^{x}\circ\mathbf{r}_{U,1},\mathbf{h}_{L}^{x}\circ\mathbf{r}_{L,1})&=(\mathbf{h}_{U}^{y},\mathbf{h}_{L}^{y}),\\
    (\mathbf{h}_{U}^{y}\circ\mathbf{r}_{U,2},\mathbf{h}_{L}^{y}\circ\mathbf{r}_{L,2})&=(\mathbf{h}_{U}^{z},\mathbf{h}_{L}^{z}),\\
    (\mathbf{h}_{U}^{x}\circ\mathbf{r}_{U,3},\mathbf{h}_{L}^{x}\circ\mathbf{r}_{L,3})&=(\mathbf{h}_{U}^{z},\mathbf{h}_{L}^{z}),
    \end{split}
\end{equation}

$\rightarrow$
\begin{equation*}
    \resizebox{\hsize}{!}{
    $(\mathbf{h}_{U}^{x}\circ\mathbf{r}_{U,1}\circ\mathbf{r}_{U,2},\mathbf{h}_{L}^{x}\circ\mathbf{r}_{L,1}\circ\mathbf{r}_{L,2}) = (\mathbf{h}_{U}^{x}\circ\mathbf{r}_{U,3},\mathbf{h}_{L}^{x}\circ\mathbf{r}_{L,3})$
    }
\end{equation*}
thus, it is inferrable that relation $r_3$ is compositional of $r_1$ and $r_2$ if and only if $\mathbf{r}_{U,1}\circ\mathbf{r}_{U,2}=\mathbf{r}_{U,3}$.
\end{proof}

\subsection{Details of Intersection Operator}
\label{Intersection Operator}
In this section, we explain the details of two important components of the intersection operator, \textbf{SemanticAverage}($\cdot$) and \textbf{CardMin}($\cdot$) 
\paragraph{SemanticAverage($\cdot$)} is expected to compute the semantic center $\pmb{\theta}_{ax}^{'}$ of the input $\{(\pmb{\theta}_{j,ax}, \pmb{\theta}_{j,ap})\}_{j=1}^{n}$. Specifically, the computation process is provided as:
\begin{equation}
\begin{split}
        [\mathbf{x};\mathbf{y}] &= \sum_{i=1}^{n} [\mathbf{a}_{j}\circ \text{cos}(\pmb{\theta}_{j,ax});\mathbf{a}_{j}\circ \text{sin}(\pmb{\theta}_{j,ax})],\\
    \pmb{\theta}_{ax}^{'} &= \textbf{Arg}(\mathbf{x},\mathbf{y}),
\end{split}
\end{equation}
where cos and sin represent element-wise cosine and sine functions. $\textbf{Arg}$($\cdot$) computes the argument given $\mathbf{x}$ and $\mathbf{y}$. $\mathbf{a}_{j} \in \mathbb{R}^d$ are attention weights such that 
\begin{equation}
    \resizebox{0.95\hsize}{!}{%
        $[\mathbf{a}_{j}]_{k} = \frac{\text{exp}([\textbf{MLP}([\pmb{\theta}_{j,ax} - \pmb{\theta}_{j,ap/2};\pmb{\theta}_{j,ax} + \pmb{\theta}_{j,ap/2}])]_{k})}{\sum_{m=1}^{n}\text{exp}([\textbf{MLP}([\pmb{\theta}_{m,ax} - \pmb{\theta}_{m,ap/2};\pmb{\theta}_{m,ax} + \pmb{\theta}_{m,ap/2}])]_{k})}, $%      
        }
\end{equation}
where \textbf{MLP} : $\mathbb{R}^{2d} \rightarrow \mathbb{R}^{d}$ is a multi-layer perceptron network, [$\cdot$;$\cdot$] is the concatenation of two vectors. 

\paragraph{CardMin($\cdot$)} predicts the aperture $\pmb{\theta}_{ap}^{'}$ of the intersection set such that $[\pmb{\theta}_{ap}^{'}]_{i}$ should be no larger than any $\pmb{\theta}_{j,ap}^{i}$, since the intersection set is the subset of all input entity sets.
\begin{equation}
    \resizebox{0.98\hsize}{!}{%
    $[\theta_{ap}^{'}]_{i} = \text{min}\{\theta_{1,ap}^{i},\dots,\theta_{n,ap}^{i}\}\cdot\sigma([\textbf{DeepSets}(\{(\pmb{\theta}_{j,ax},\pmb{\theta}_{j,ap})\}_{j=1}^{n})]_{i})$%
    }
\end{equation}
where \textbf{DeepSets}$(\{(\pmb{\theta}_{j,ax}, \pmb{\theta}_{j,ap})\}_{j=1}^{n})$ \citep{deepsets} is given by
\begin{equation*}
    \resizebox{0.95\hsize}{!}{%
    $\textbf{MLP}(\frac{1}{n}\sum_{j=1}^{n}\textbf{MLP}([\pmb{\theta}_{j,ax} - \pmb{\theta}_{j,ap/2};\pmb{\theta}_{j,ax} + \pmb{\theta}_{j,ap/2}])).$%
    }
\end{equation*}

\subsection{Distance Function}
\label{Distance function}
Inspired by \citet{ConE,Query2box}, the distance between $\mathbf{q}$ and $\mathbf{h}^{*}$ is defined as a combination of inside and outside distances, $d_{com}(\mathbf{q}, \mathbf{h}^{*}) = d_{o} + \lambda d_{i}$:
\begin{equation}
\begin{split}
    d_{o}(\mathbf{q}, \mathbf{h}^{*}) &= min\{\|\mathbf{h}_{U}^{q} - \mathbf{h}^{*}\|_{1}, \| \mathbf{h}_{L}^{q} - \mathbf{h}^{*}\|_{1}\},\\
    d_{i}(\mathbf{q}, \mathbf{h}^{*}) &= min\{\|\mathbf{h}_{ax}^{q} - \mathbf{h}^{*}\|_{1}, \|\mathbf{h}_{U}^{q} - \mathbf{h}_{ax}^{q}\|_{1}\},
\end{split}
\end{equation}
where $\|\cdot\|_{1}$ is the $L_1$ norm, $\mathbf{h}_{ax}^{q}$ represents the cone center, and $\lambda \in (0,1)$.\\
Note that $d_{com}$ can only be used for measuring the distance between a single query embedding and an answer entity vector. Since we represent the disjunctive queries in Disjunctive Normal Form as a set of query embeddings $\{\mathbf{q}_{1},\dots,\mathbf{q}_{n}\}$, the distance between the answer vector and such set of embeddings is the minimum distance between the vector and each of these sets:
\begin{equation}
    \resizebox{\hsize}{!}{
    $d_{dis}(\mathbf{q},\mathbf{h}^{*}) = min\{d_{com}(\mathbf{q}_{1}, \mathbf{h}^{*}),\dots,d_{com}(\mathbf{q}_{n}, \mathbf{h}^{*})\}
    $}
\end{equation}
Thus, the overall distance function can be defined as:
\begin{equation}
\resizebox{\hsize}{!}{
  $d(\mathbf{q},\mathbf{h}^{*}) =
    \begin{cases}
      d_{com}(\mathbf{q},\mathbf{h}^{*}), & \text{query without disjunction}\\
      d_{dis}(\mathbf{q},\mathbf{h}^{*}), & \text{query with disjunction}
    \end{cases}
    $}
\end{equation}

\subsection{More Details about Experiments}
\label{Experiment}
In this section, we elaborate more details about our experiments. The code is anonymously available at \url{ https://anonymous.4open.science/r/RoConE-1C1E}.
\subsubsection{Dataset and Query structures}
RoConE is compared with various state-of-the-art query embedding models, including GQE \citep{GQE}, Query2Box \citep{Query2box}, BetaE \citep{BetaE}, and ConE \citep{ConE}. For a fair comparison with existing query embedding models, we use the same query structures and datasets NELL995 \citep{deeppath} and FB15k-237, and the open-sourced framework created by \cite{BetaE} for logical query answering tasks. Table \ref{Dataset Statistics} summarizes the descriptive statistics about the benchmark datasets. Figure \ref{query structures} illustrates all query structures used in our experiments. 

\begin{table*}
\centering
\begin{tabular}{lllllll}
\hline
\textbf{Dataset} & \multicolumn{2}{c}{\textbf{Training}} & \multicolumn{2}{c}{\textbf{Validation}}& \multicolumn{2}{c}{\textbf{Testing}}\\
\hline
& \textbf{EPFOL} & \textbf{Negation} & \hspace{3mm}\textbf{1p} & \textbf{others} & \hspace{3mm}\textbf{1p} & \textbf{others}\\ 
FB15k-237 & 149,689 & 14,968 & 20,101 & 5,000 & 22,812 & 5,000 \\
NELL995 & 107,982 & 10,798 & 16,927 & 4,000 & 17,034 & 4,000 \\
\hline
\end{tabular}
\caption{\label{Dataset Statistics}
Statistics of benchmark datasets : FB15k-237 and NELL995. EPFOL represents queries without negation, namely 1p/2p/3p/2i/3i.}
\end{table*}

\begin{figure*}
    \centering
    \includegraphics[scale=0.7]{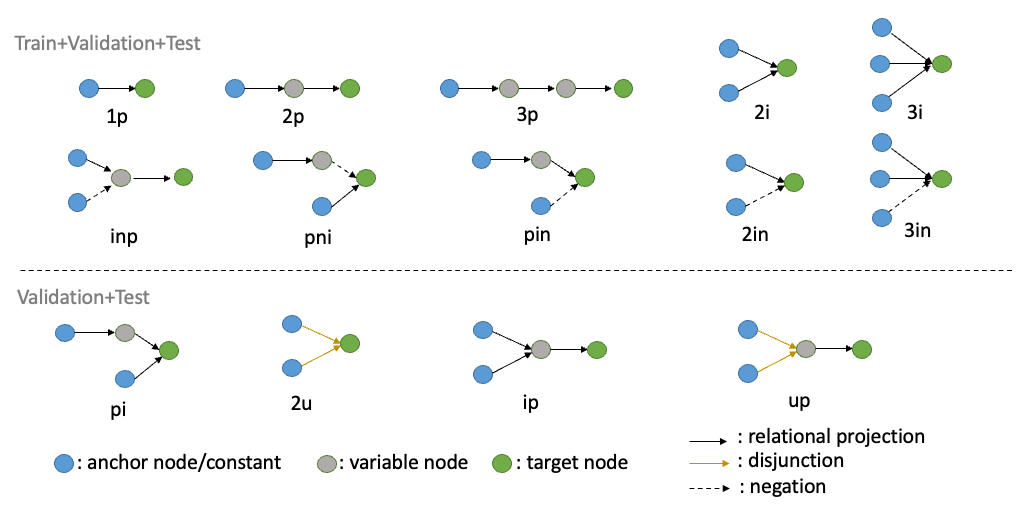}
    \caption{Fourteen types of queries used in the experiments. "$p$" represents relational projection, "$i$" represents intersection, "$u$" represents union and "n" represents negation. The upper of queries is used in the training stage, while all queries are evaluated in the validation and test stages.}
    \label{query structures}
\end{figure*}

\subsubsection{Hyperparameters and Computational Resources}
All of our experiments are implemented in Pytorch \citep{DBLP:journals/corr/abs-1912-01703} framework and run on four Nvidia A100 GPU cards. For hyperparameters search, we performed a grid search of learning rates in $\{5\times 10^{-5},10^{-4},5\times10^{-4}\}$, the batch size in $\{256, 512, 1024\}$, the negative sample sizes in $\{128,64\}$, the regularization coefficient $\lambda$ in $\{0.02, 0.05, 0.08, 0.1\}$ and the margin $\gamma$ in $\{20, 30, 40, 50\}$. The best hyperparameters are saved in Table \ref{Hyperparameters}. 

\begin{table}
\centering
\resizebox{\hsize}{!}{
\begin{tabular}{lllllll}
\hline
\textbf{Dataset} & d & b & n & $\gamma$ & $\hspace{4mm} l$ & $\lambda$\\
\hline
FB15k-237 & 1600 & 128 & 512 & 30 & 5 $\times 10^{-5} $ & 0.1\\
NELL995 & 800 & 128 & 512 & 20 & $1 \times 10^{-4}$ & 0.02\\
\hline
\end{tabular}
}
\caption{\label{Hyperparameters}
Hyperparameters found by grid search. d represents the embedding dimension, b is the batch size, n is the negative sampling size, $\gamma$ is the margin in loss, l represents the learning rate, $\lambda$ is the regularization parameter in the distance function.
}
\end{table}

\subsubsection{Evaluation Metrics}
In this paper, we use Mean Reciprocal Rank (MRR) as the evaluation metric. Given a sample of queries Q, Mean reciprocal rank represents the average of the reciprocal ranks of results:
\[MRR = \frac{1}{|Q|}\sum_{i=1}^{|Q|}\frac{1}{rank_{i}}\]

\subsubsection{Variants of relational projection strategies}
\label{Variants of RoConE}
To better incorporate the learning the relation patterns in query embedding model, we proposed alternative two variants of the relational rotation projection strategy:
\begin{itemize}
    \item RoConE (Trunc): after relational rotation, the two boundaries of query embedding are truncated to ensure that aperture will not be greater than $\pi$.  
    
    \item RoConE (S.E): replaces relational rotation on boundaries by a non-linear sigmoid function which maps the boundaries to shrink or expand and only maintains rotations on $\theta_{ax}$.
\end{itemize}

\subsubsection{Error Bars of Main Results}
We have run RoConE 10 times with different random seeds, and obtain mean values and standard deviations of RoConE’s MRR results on EPFO and negation queries. Table \ref{EPFOL(sd)} shows that mean values and standard deviations of RoConE’s MRR results on EPFOL queries. Table \ref{negation(sd)} shows mean values and standard deviations of RoConE’s MRR results on negation queries.
\begin{table}
\centering
\resizebox{\hsize}{!}{
\begin{tabular}{llllllllll}
\hline
\textbf{Dataset} & \textbf{1p} & \textbf{2p} & \textbf{3p} & \textbf{2i} & \textbf{3i} & \textbf{pi} & \textbf{ip} & \textbf{2u} & \textbf{up}\\
\hline
$\multirow{2}{5em}{FB15k-237}$ & 42.2$\pm$ & 10.5$\pm$ & 7.5$\pm$ & 33.5$\pm$ & 48.1$\pm$ & 23.5$\pm$ & 14.5$\pm$ & 12.8$\pm$ & 8.9$\pm$\\
& 0.054 & 0.108 & 0.142 & 0.075 & 0.172 & 0.105 & 0.183 & 0.096 & 0.193\\
\hline
$\multirow{2}{5em}{NELL995}$ & 54.5$\pm$ & 17.7$\pm$ & 14.4$\pm$ & 41.9$\pm$ & 53.0$\pm$ & 26.1$\pm$ & 20.7$\pm$ & 16.5$\pm$ & 12.8$\pm$\\
& 0.097 & 0.152 & 0.192 & 0.130 & 0.038 & 0.084 & 0.146 & 0.112 & 0.137\\
\hline
\end{tabular}
}
\caption{RoConE's MRR mean values and standard variances ($\%$) on answering EPFO ($\exists$, $\land$, $\lor$) queries}
\label{EPFOL(sd)}
\end{table}

\begin{table}
\centering
\resizebox{\hsize}{!}{
\begin{tabular}{llllll}
\hline
 \textbf{Dataset} & \textbf{2in} & \textbf{3in} & \textbf{inp} & \textbf{pin} & \textbf{pni}\\
\hline
$\multirow{2}{5em}{FB15k-237}$ & 4.1$\pm$ & 7.9$\pm$ & 6.9$\pm$ & 3.1$\pm$ & 2.8$\pm$\\
 & 0.089 & 0.092 & 0.161 & 0.083 & 0.077\\
\hline
$\multirow{2}{5em}{NELL995}$ & 5.2$\pm$ & 7.7$\pm$ & 9.4$\pm$ & 3.2$\pm$ & 3.7$\pm$\\
& 0.012 & 0.079 & 0.154 & 0.012 & 0.097\\
\hline
\end{tabular}
}
\caption{RoConE's MRR mean values and standard variances ($\%$) on answering negation queries}
\label{negation(sd)}
\end{table}

\end{document}